\documentclass[11pt, english]{article}

\usepackage{fullpage}
\usepackage{amsmath}
\usepackage{amssymb}
\usepackage{amsthm}
\usepackage{pxfonts}
\usepackage[T1]{fontenc}
\usepackage{color}
\usepackage{graphicx}
\usepackage{float}
\usepackage{pdfsync}
\usepackage{hyperref}

\usepackage{charter}
\usepackage{xspace}

\newtheorem{theorem}{Theorem}
\newtheorem{lemma}[theorem]{Lemma}

\newtheorem{corollary}[theorem]{Corollary}

\newcommand{\maxdeg}{\ensuremath{maxdeg}\xspace}
\newcommand{\degbound}{\ensuremath{\Delta}\xspace}
\newcommand{\dist}{\ensuremath{\lambda}\xspace}
\newcommand{\const}{\ensuremath{\gamma}\xspace}
\newcommand{\ds}{\ensuremath{\mathcal D}\xspace}
\newcommand{\gvd}{\ensuremath{\mathcal V}\xspace}
\newcommand{\for}{\ensuremath{\mathcal F}\xspace}
\newcommand{\lvt}[1]{\ensuremath{LVT(#1)}\xspace}
\newcommand{\clvt}[1]{\ensuremath{cLVT(#1)}\xspace}
\newcommand{\cov}[1]{\ensuremath{cOV(#1)}\xspace}
\newcommand{\ov}[1]{\ensuremath{OV(#1)}\xspace}
\newcommand{\rep}[1]{\ensuremath{rep(#1)}\xspace}
\newcommand{\repe}[2]{\ensuremath{rep_{(#1)}(#2)}\xspace}
\definecolor{ltg}{rgb}{0.9,0.9,0.9}

\definecolor{SEB}{rgb}{0,0,1}
\definecolor{JOHN}{rgb}{0.42,0,0.87}

\title{Confluent Persistence Revisited}

\date{}

\author{S\'{e}bastien Collette\thanks{Chargé de recherches du F.R.S.-FNRS. Université Libre de Bruxelles, Belgium. secollet@ulb.ac.be} \and John Iacono%
\thanks{Department of Computer Science and Engineering,
Polytechnic Institute of New York University, 6 MetroTech Center, Brooklyn NY 11201 USA. http://john.poly.edu. Research supported by NSF grants CCF-1018370, CCF-0430849 and by an Alfred P.~Sloan fellowship. Research partially completed as a visiting professor at MADALGO (Center for Massive Data Algorithmics, a Center of the Danish National Research Foundation), Department of Computer Science, Aarhus University, IT Parken, \r{A}bogade 34, DK-8200 \r{A}rhus N, Denmark.} 
\and 
Stefan Langerman%
\thanks{Maître de recherches du F.R.S.-FNRS. Université Libre de Bruxelles, Belgium. slanger@ulb.ac.be}}

\begin{document}
\maketitle

\begin{abstract} \noindent
It is shown how to enhance any data structure in the pointer model to make it confluently persistent, with efficient query and update times and limited space overhead. Updates are performed in $O(\log n)$ amortized time, and following a pointer takes $O(\log c \log n)$ time where $c$ is the in-degree of a node in the data structure. In particular, this proves that confluent persistence can be achieved at a logarithmic cost in the bounded in-degree model used widely in previous work. This is a $O(n/\log n)$-factor improvement over the previous known transform to make a data structure confluently persistent.
\end{abstract}

\section{Introduction}

\paragraph{Persistence.}
Abstractly, a \emph{data structure} is some collection of data that supports various \emph{operations} to be performed upon it. These operations can be partitioned into two groups: \emph{query} operations, that do not modify the structure, and \emph{updates} that do modify the structure. Classically, operations are performed on the current state of the structure, and in the case of an update, the current state is changed, and the past state is lost. 

\emph{Persistence} is the concept, first formally studied in~\cite{FullPersistence}, of allowing the efficient execution of operations not just on the current state, but on past states as well. The past states are referred to as \emph{versions} of the structure. There are three forms of persistence, which we discuss in increasing order of power and complexity.

The simplest kind, \emph{partial} persistence, is where one can execute query operations on any previous version of the structure, but where updates are only allowed on the current version. 
This implies that the versions form a list (or timeline), each version being derived from the latest (or most recent) version in the list.

In \emph{full} persistence, update and query operations may be performed on any version of the data structure that ever existed. In the case of an update on a past version, a new version is created, derived from that past version. While in the classic sense, and in partial persistence, time takes on its normal linear structure, in full persistence, time as viewed by the data structure no longer has a linear topology, it has become tree-like. The \emph{version tree} is the graph where every node is a version, and an edge is drawn from every version to the version it was derived from.

In what has been presented so far, updates work on a single version of the structure to produce a new version. One type of operation common in data structures is a \emph{merge} or \emph{meld} operation that takes two logically independent components of a data structure and merges them into one. Such operations have been supported so far so long as the components lie within the same version. However, since full persistence gives a rich tree topology of versions, this opens the possibility of performing a merge-type operation with components coming from different versions. 
Such update operations, that take as input two (or more) versions and produce a single new version, are called \emph{confluent updates}; full persistence with the support of confluent updates is called \emph{confluent persistence}, and was introduced in \cite{PersistentLists, FullPersistence}. In confluent persistence, since each version is derived from one or more previous versions, the graph formed by the versions derivation relationship is not longer a tree, it is now the \emph{version DAG}.

\paragraph{Applications of persistence.} 
Persistence is a very useful tool, with applications ranging from algorithm design to code debugging. We list some of these here.  

From a theoretical point of view, persistence has played a major role in the development of efficient algorithms. For instance, the best algorithms for planar point location~\cite{PointLocPersistence} and ray-shooting~\cite{kaplan2009linear} use persistence. 

A widely-used technique in computational geometry consists in performing a sweep with a vertical line through a data set (set of points, set of segments, line arrangement, map, data structure...), and storing in a data structure the set of items intersected by the line. As we sweep, we update the data structure to accurately represent the objects intersected by the sweep line.  If the data structure is made persistent, one can then retrieve all the information that has been swept through, by going back to previous versions. For instance, this can be used to solve optimally the segment intersection problem~\cite{persistenceCG}.

Practically, researchers have developed generic tools, implemented in various object-oriented programming languages, to convert any data structure into a partially or fully persistent one by simply adding one keyword to its source code~\cite{PersistenceImplemented}. In other words, persistence can be obtained easily from a developer's point of view, which further motivates the development of algorithms making use of persistent structures. 

Finally, confluent persistence is useful in many contexts where one wants to backup all versions of some information. Typical examples include version control systems~\cite{ConfluentTries_Algorithmica}, such as Subversion and CVS. In these systems, all the versions of some files are stored, in order to be able to go back to any previous version of any file. A new version can be obtained by combining files from previous versions. Similarly, code debugging (e.g.,~\cite{DebugPersistence}) can greatly benefit from the use of persistent data structures: by storing all the versions of the variables over the lifetime of a program, one can check, at any point in time, what the value of any variable was, and how it was modified by the program. 

\paragraph{Previous Results.}
Many results were obtained by considering specific data structures. For instance, ad-hoc solutions for fully persistent stacks and lists\cite{PersistentLists,PersistentLists2}, arrays\cite{PersistentArrays}, or confluently persistent tries\cite{ConfluentTries} and deques\cite{ConfluentDeques} have been proposed, among many others. Here, we focus on generic methods, which can be implemented transparently, and add support of persistent operations to any data structure in the pointer model. Such support can be added trivially by copying the structure before performing any update, so the question arises as to what is the most efficient way, in terms of both space and time, to support persistent operations.

In that direction, solutions are known since 1986, when Driscoll\emph{~et al.}~\cite{FullPersistence} proposed an efficient method to transform any data structure in the pointer model with constant in-degree into a partially or fully persistent one. The space usage and time overhead of their construction is asymptotically optimal: the space used is the number of total changes in the structure, while operations are slowed down by only a constant factor vis-\`{a}-vis their non-persistent counterparts. 

In their seminal paper, Driscoll\emph{~et al.}~mentioned confluent persistence as an open problem; it was formally defined in 1994~\cite{PersistentLists}. And despite many publications on persistence, the problem remained mostly open until 2003, when Fiat and Kaplan~\cite{ConfluentPersistence} gave the first generic method to obtain confluently persistent structures.

As pointed out by Fiat and Kaplan, there is a fundamental problem to overcome to get efficient confluent persistence: a merge of two previous versions can potentially double the size of the data structure, and therefore in $k$ merges, one can get a data structure of size $2^{\Omega(k)}$. Based on this, they provide a lower bound of $\Omega(U \, e(D))$ bits to store all versions where $U$ is the number of atomic changes performed over the whole history of the structure and $e(D)$ is the so-called ``effective" depth of the version DAG, defined roughly as the logarithm of the maximum number of different paths from the root of the DAG to each of its versions. 

In terms of upper bounds, the solution proposed in~\cite{ConfluentPersistence} has a space complexity of $O(U \, e(D) \log U)$ bits and an overhead of $O(e(D) + \log U)$ time per operation on the enhanced data structure. However, since the effective depth can be linear in the number of updates performed, these time and space bounds could be terrible, and in fact in the worst case would not be an improvement over simply copying the whole structure to create each new version.

\paragraph{Our Work.}
To overcome this issue of exponential structures, we tackle the problem in a different way: what if we want to allow confluent persistent operations on a pointer-model structure, but we never use two versions of the \emph{same} node in a confluent merge? We will show in the next sections that with this extra restriction, confluent persistence can be achieved very efficiently, using $\Theta(U)$ space\footnote{If words are assumed to have $\log U$ bits, and assignment of a word takes $O(1)$ time, our approach uses $O(U \log U)$ bits, using the notation of~\cite{ConfluentPersistence}.}, and $O(\log U)$ time to implement a unit-cost pointer-model operation, when the data structure to enhance has constant in-degree. 

This restriction on merging versions of the same node is natural. For instance, for confluent version control system, such as Subversion, it would mean that we do not allow to create a new version with two separate copies of the same file. For confluent text editors, we would not be allowed to merge two copies of the same paragraph of text, but we could merge a section from version $a$ with another section from version $b$. Furthermore, this restriction can be overcome, by explicitly duplicating the nodes, which then allows us to merge \emph{copies} of different versions of the same node. In other words, merging parts of the structure with themselves can be done, but at a cost linear in their size. 

Given the prohibition on self-merging, our data structure is a significant improvement over that of~\cite{ConfluentPersistence} when confluent updates are numerous (and supporting confluent updates is the whole point of confluent persistence!). For example, consider a data structure that maintains strings and allows split and concatenate operations, and queries as to the $i$th character of a given string; this can be done using, for example, red-black trees in $O(\log n)$ time. Concatenation is the confluent update; so as long as we do not allow concatenation with a previous version of oneself the requirements of our structure are met. If we start with a single string of length $n$ and perform $n$ operations, our total runtime will be $O(n \log^2 n)$, whereas using the structure of Fiat and Kaplan, the runtime bound is $O(n^2 \log n)$, a substantial difference.

\paragraph{Overview of approach.}
The main idea used here is to maintain, in each node of the data structure, and associated to each pointer of the structure, a compressed representation of (a tree of) the versions where that node or pointer existed. In other words, each object will contain some compacted representation of its own history. Using a combination of fractional cascading and bi-directional pointers, we will ensure that one can efficiently browse a given version of the structure and know what the state of each node or pointer was in that specific version. 

Out of the nodes of the data structure, we also maintain the full version DAG, as well as a carefully chosen collection of  disjoint trees -- or forest -- covering all versions. This forest will be stored in a link-cut tree. Using this forest, we will be able to efficiently determine if a version was derived from another, which is essential to navigate in the structure. This is the key to avoid the use of order maintenance structures, as used for full persistence, which are not powerful enough to support confluent operations. 

\section{Model}
\paragraph{Computational Model}Our results belong to the pointer model. Let \ds be a data structure which shall be made confluently persistent. The data structure $\mathcal{D}$ is required to be a \emph{linked data structure} as defined in Driscoll\emph{~et al.~}\cite{FullPersistence}: it is composed of a set of \emph{nodes}, each node contains a bounded number of fields, which can either contain values (integers, booleans, \emph{etc.}) or pointers. We refer to pointers of the data structure $\mathcal{D}$ as \emph{edges}; these edges are \emph{directed}, from a \emph{predecessor} to a \emph{successor}. 
The atomic operations that we can perform on a data structure in this model are to \emph{create} or \emph{delete} nodes, to follow an edge, or to view or change  value of any field of a node (including edges). Changing one field is stipulated to cost $O(1)$ time. 

Each structure has a unique node designated as the \emph{root}. In a regular operation, one root is specified and a sequence of atomic operations is performed. The other type of operation is a \emph{merge} operation. In a merge operation, two roots are specified and at the end of the operation the second root is no longer considered to be a root; this implies that the connected components of the two roots have been merged. While temporary pointers are allowed during the execution of an operation, no pointers into the structure are permitted to be maintained between operations other than to the roots.

\paragraph{In-degree} Our transformation to make $\mathcal{D}$ confluently persistent requires a constant bound on the incoming degree of every node of $\mathcal{D}$. The bound will be denoted by \degbound. Notice that if the incoming degree of a data structure is not bounded by a given constant, we can easily transform the data structure in such a way that every node will have a bounded in-degree: we replace each node with high in-degree by a balanced binary search tree, of depth $O(\log \maxdeg)$, where \maxdeg is the maximum in-degree of a node. Then, instead of having pointers pointing directly to the node, the pointers lead to the leaves of the balanced binary search tree. 

Asymptotically, this transformation is made at no additional space cost (as the number of constant-sized nodes added is linear in the number of edges). But the cost of following a pointer is then $O(\log \maxdeg)$, and the other operations are still constant. As we require that all outgoing and incoming edges of a node are deleted before deleting the node itself, creating and deleting a node is also done at no additional time cost in the modified data structure. This in-degree restriction is common to the previous work on persistence (see for instance \cite{FullPersistence}).

\paragraph{Confluent Persistence} To make a data structure confluently persistent is to augment the data structure with the following temporal operations: given any version, execute a regular operation, creating a new version for any changes that are made; execute a merge operation which combines the information contained in two versions to obtain a new version.

To keep the analysis simple, we will consider, in what follows, that the versions in the version DAG have out-degree at most 2. This is not a constraint, because if we want to fork multiple versions from a single version, we can always do so by adding a constant number of dummy version nodes. 

Finally, as pointed out in the introduction, we have to avoid the case where a version is obtained by merging two previous versions containing the same node, as this operation could lead to a data structure of size $\Omega(2^{k})$ in $k$ operations. In this work, we consider that such a merge will never occur: we forbid the user to merge two versions containing common nodes. It is the the responsibility of the user to ensure that common nodes do not exist in versions to be merged\footnote{This restriction can be circumvented by replacing one of the two identical nodes with a duplicate before executing a merge operation. This will have the effect of increasing space usage to the actual size of the data structure being represented.}. 
In practice, a merge will consist in adding an edge between two nodes belonging to two different versions, with the effect of connecting two disjoint components. 

\section{Augmenting the data structure}
In what follows we will show how a data structure \ds is augmented, and associated with local and global version DAGs. A first thing to notice is that when we refer to \ds, we mean the collection of all the nodes in all versions. The extra information added to \ds will let us differentiate between the nodes which are valid or not in any given version, and will let us determine what is the information contained in the fields of a node for a specific version. 

\subsection{Global Version DAG}
We maintain a global version DAG \gvd, associated to \ds. Each time a new version is created, a new vertex representing it is added to \gvd. Except for the root, each version is the child of one or two versions, depending if the version is the result of a confluent merge, or if it is simply a branch. When a version is created, it receives an ID, which is an integer larger than any other version ID in the DAG; in other words we maintain a version counter that we increment each time a new version is created. This ensures that the DAG has the min-heap property: the ID of each version is strictly larger than the version ID of its parent(s), and each path from the root is monotonically increasing. 

We maintain bi-directional pointers so that we can, from each version, determine from what version it was derived.
The purpose of \gvd is to give access to the root of each version; therefore each version $v$ of \gvd contains a pointer to the node in \ds which is the root of \ds in version $v$. 

\begin{figure}
\begin{center}
\includegraphics[scale=.7]{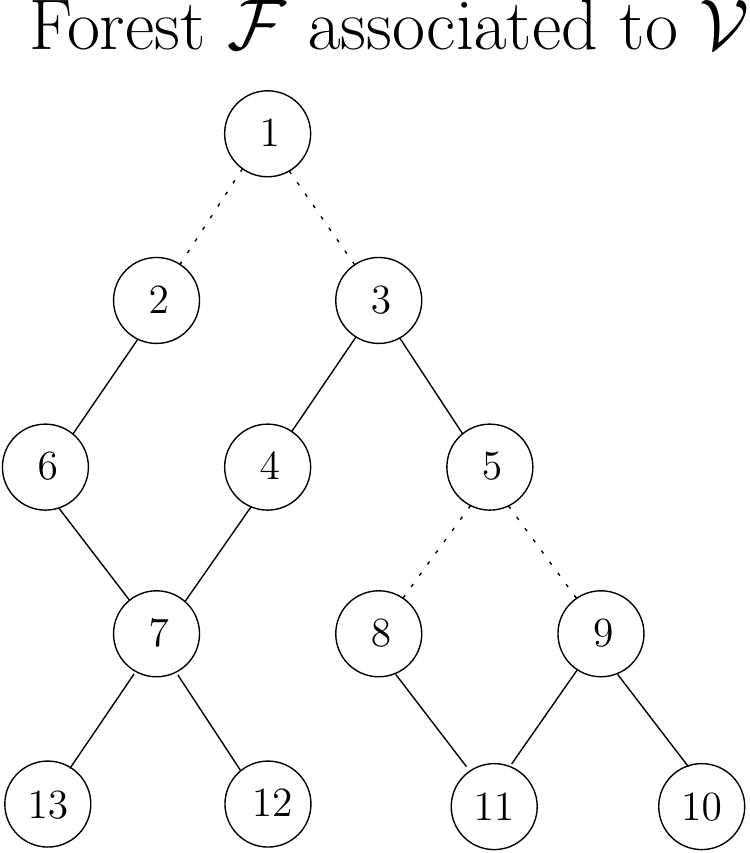}
\end{center}
\caption{\label{fig:subtree}The forest defined by the plain edges is stored as a link-cut tree.}
\end{figure}

We also maintain a secondary representation of a subgraph of $\gvd$: we maintain a collection of disjoint trees, covering all the nodes of $\gvd$ and containing a subset of its edges. When a version is added to $\gvd$, it is also added to the collection of trees and adjacent to the same previous version. 
This forest, denoted by $\for$, will be stored in a link-cut tree~\cite{LinkCutTree}, with bi-directional pointers between the corresponding versions in $\for$ and $\gvd$. Whenever a confluent merge occurs, we determine, in $\for$, what is the lowest common ancestor of the two versions that we merge. And we delete (or cut) the two edges connecting the lowest common ancestor to its descendants, as shown on Figure~\ref{fig:subtree}. This ensures that no cycle will ever be created, and therefore each component of $\for$ is indeed a combinatorial tree and $\for$ is a forest. Note that in the case where no lowest common ancestor exists (that is, when the two version merged already belong to two different trees in $\for$), we do not cut any edge, and the invariant that the resulting graph is a tree is indeed preserved. 

To find a lowest common ancestor in a link-cut tree between node $v$ and node $v'$, we store the versions with their version ID as the key. Link-cut trees allow to find the lowest ID on any path in $O(log n)$ time, and allow to check if two nodes belong to the same tree of the forest in $O(log n)$ time~\cite{LinkCutTree}. Therefore by querying the path $vv'$, we get the version with lowest ID on the path, which turns out to be the lowest common ancestor of $v$ and $v'$, or we determine that $v$ and $v'$ are part of disjoint components. 

Thus, initially, when \gvd is created, it contains a single version. When we add new versions to \gvd, we also insert them in the link-cut tree in $O(\log n)$ amortized time. If a version is created by a merge, it might require two cuts in \for, done in $O(\log n)$ amortized time. 

\subsection{Local Version Tree}
Each node $d$ of \ds is present in a certain number of versions, we denote the subgraph of \gvd where $d$ exists by the \emph{local version DAG} of $d$.  Because of the requirement that two versions can only be merged if they contain disjoint sets of nodes, local version DAG of $d$ is, in fact, a tree:
\begin{lemma}
The local version DAG of a node is a tree. 
\end{lemma}

\begin{proof}
First, notice that the set of versions of \gvd where a given node $d$ exists is connected: each time a new version $v$ is created which contains $d$, $v$ must be the child of a version containing $d$. Second, imagine that the local version DAG of $d$ is not a tree: then there exists a version $v$ containing $d$, which has two (or more) incoming edges, from two versions $v_{1}$ and $v_{2}$ containing $d$. But this contradicts the requirement that merging two versions containing the same node is forbidden.
\end{proof}

\begin{figure}
\begin{center}
\includegraphics[scale=.7]{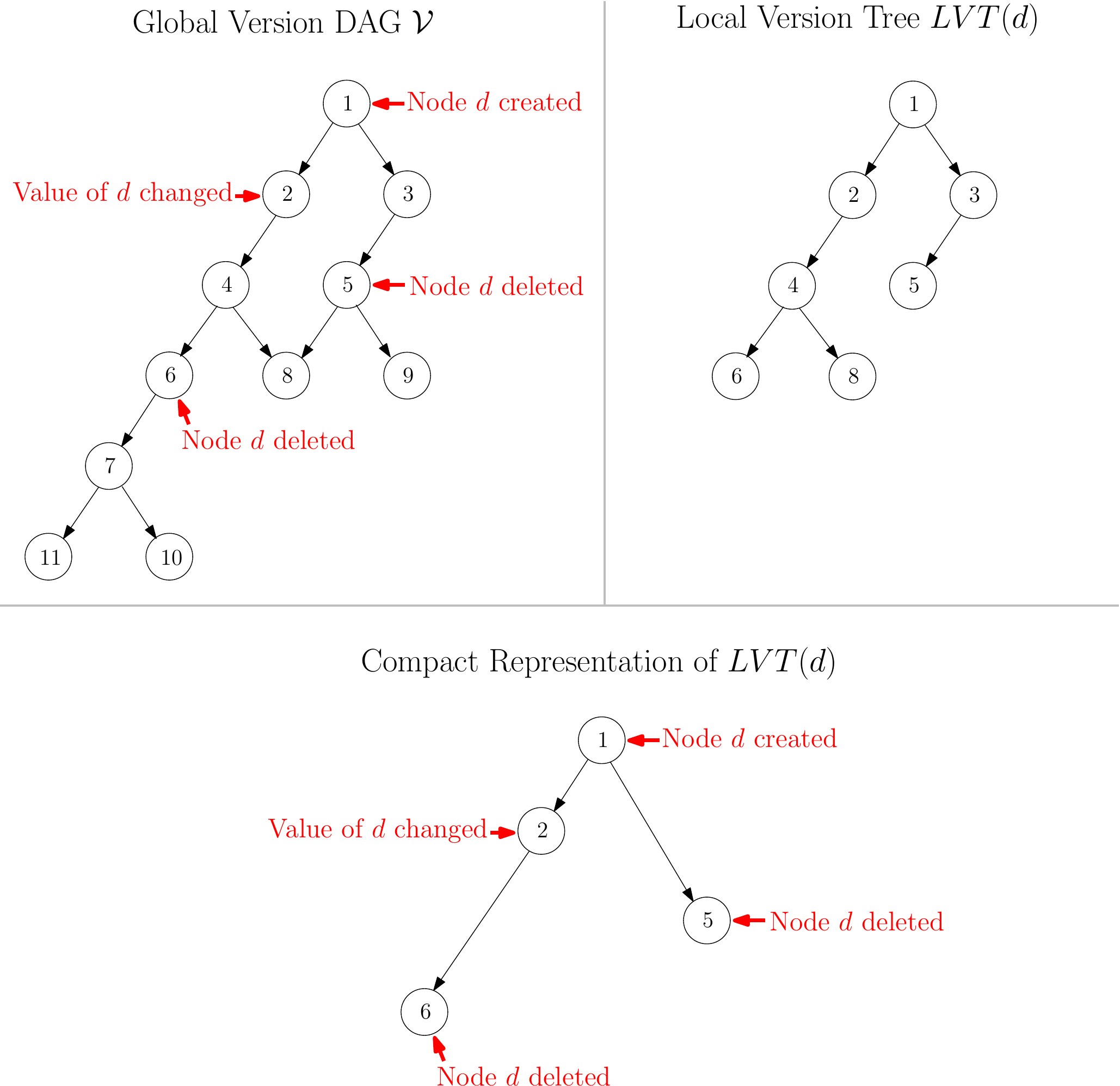}
\end{center}
\caption{\label{fig:compact}Global version DAG, local version tree and its compact representation.}
\end{figure}

Thus, henceforth the local version DAG of a node $d$ shall be known as its \emph{local version tree} and denoted as \lvt d. For each node $d$, the local version tree is not explicitly stored. This would take too much time and space. 
What is stored instead, is a compressed version of the \lvt d called the \emph{compressed local version tree} and denoted as \clvt d. Conceptually, \clvt d is obtained from \lvt d by contracting with their parent all vertices corresponding to versions where none of the fields of $d$ were changed. (Note that the root of \lvt t will not be contracted as this is the version where node $d$ came into existence and its fields were initialized). Thus, the vertices of \clvt d correspond to maximally connected subtrees of \lvt d where all the fields in $d$ are unchanged and identical in all the versions corresponding to each vertex of each subtrees.

A vertex of \lvt d is \emph{explicit} if it is in \clvt d; otherwise it is termed \emph{implicit}. An implicit node is \emph{represented by} the node in the compressed local version tree that it was contracted to; the representative of $v$ is denoted by $\rep v$. The fields of node $d$ in a implicit version $v$ in \lvt d are identical to the fields of the node in $d$ in version \rep v.

As previously mentioned, for each node $d$ of the data structure \clvt d is explicitly stored. Associated with each vertex of this compressed local version tree, all the fields of the node in the version corresponding to the vertex  are stored. 

\subsection{Edge Overlays}
Let $a,b$ be two nodes in \ds, such that there is an edge $e$ from $a$ to $b$ in some version $v$. Notice that the set of versions for which the edge exists is a connected subset of \lvt{a} and \lvt{b}, which is defined to be \ov{a,b}; the representatives of this set of versions is a connected subset of \clvt a and \clvt b, the intersection of which which is defined to be \cov{a,b}. For an edge $(a,b)$ that exists at some version $v$, we call the lowest ancestor of $v$ in \lvt{a} (or \clvt a) that is in \ov{a,b} the representative of $v$ with respect to edge $(a,b)$ and denote this as \repe{a,b}{v}.

\begin{figure}
\begin{center}
\includegraphics[scale=.8]{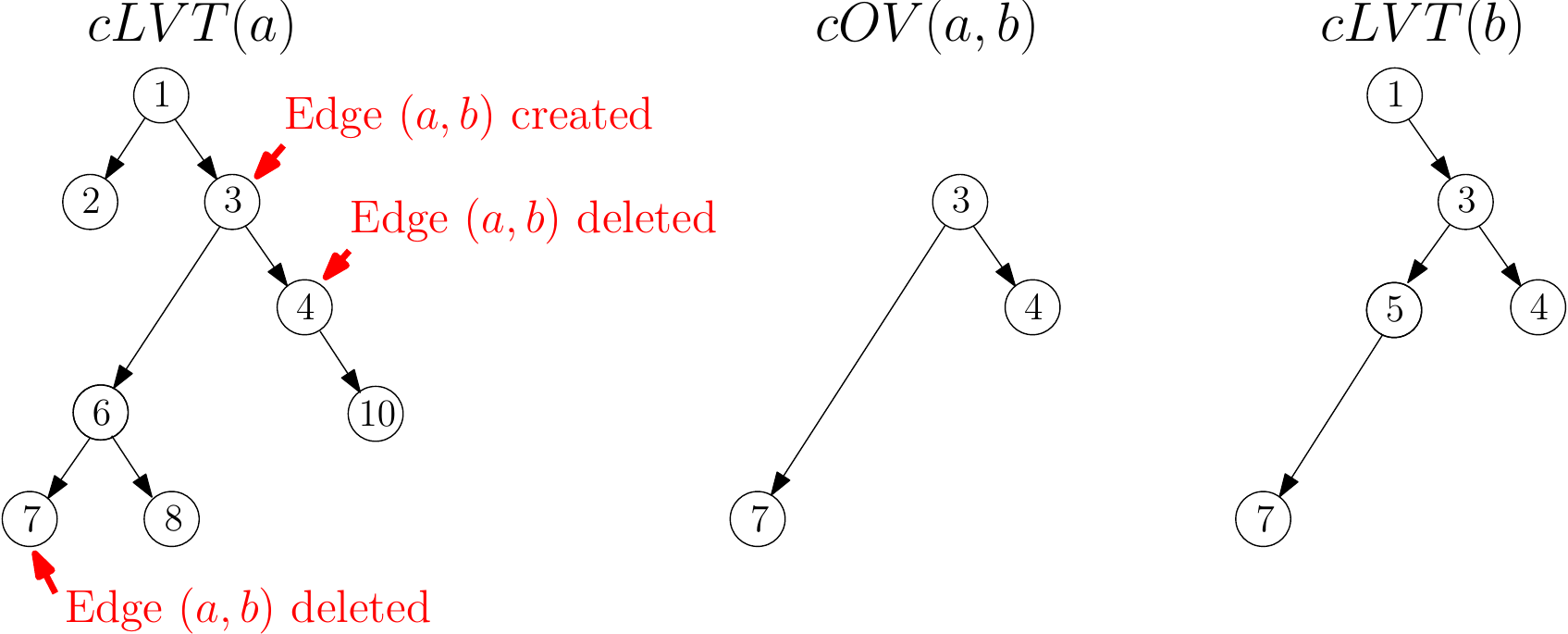}
\end{center}
\caption{\label{fig:overlay}Compact representation of the overlay tree for the edge $(a,b)$.}
\end{figure}

For each edge $e$, we store the compressed representation of the overlay tree \cov{a,b}  (see Figure~\ref{fig:overlay}). 
We maintain the following invariants, by adding explicit nodes when required: \cov{a,b} must always contain an explicit vertex for each version were the edge has been created or deleted; and each explicit vertex of the contracted overlay must explicitly appear in both \clvt{a} and \clvt{b} (but not necessarily reversely).

\subsection{Node-Edge Navigation Structure}

The compressed version tree structure representing versions of nodes and the compressed  edge overlay structure representing versions of edges have been defined. Now what is needed is some additional structure to help navigate from the versions of one structure to the other and vice-versa.
To this effect several categories of pointers which will be stored are defined; these constitute an aspect of the structure called the \emph{node-edge navigation structure}: 

\begin{itemize}
\item \emph{Version-edge pointers.}
For each edge $(a,b)$ that exists in a version $v$ in \clvt{a},  add pointers from
$v$ in \clvt{a} to \repe{a,b}{v} in \cov{a,b}. 

\item \emph{Edge-version pointers.} These are pointers back from all versions $v$ in \cov{a,b} to the same version $v$ in \clvt{a}.

\item \emph{Inverse edge-version pointers.} A list of reverse pointers, so that from a version in $\clvt{b}$ we can retrieve the list of its incoming edge-version pointers.    
\end{itemize}

These pointers in the node-edge navigation structure properly associate nodes and edges, but navigation between the two structures is not always efficient since there could, for example, be a path of many vertices in a compressed version tree of a node without a pointer in the node-edge navigation structure to a vertex in a compressed overlay versions of an incident edge. For efficiency reasons, storing all such pointers is not possible; however we employ the technique of fractional cascading~\cite{FracCasc} in order to allow efficient navigation.

For each node $d$ of \ds, a certain number of vertices appearing in $\clvt{d}$ will be \emph{senators}. We will ensure to pick as senator at most a fraction $1/{\dist}$ of the vertices on each path (as will be shown later, $\dist$ is a small constant). A version $v$ is a senator when it has an ancestor $v'$, at distance $\dist$, such that no node on the path from $v$ to $v'$ is senatorial. We maintain senators as new versions are added to the local version tree, by going up in the tree at most $\dist$ steps. 

When a version $v$ is added to \clvt{d}, if it is a senator, we will add this version $v$ explicitly in all its incoming edge overlays, as well as in \clvt{d'} for all nodes $d'$ of \ds for which there is an edge from $d'$ to $d$ at version $v$. These newly inserted explicit versions might also be senators, and in that case we will propagate senators to their neighbors and so on. We will show the amortized cost of this operation in the analysis section. 

At the end of this process, we end up with local version trees in which there is always a senator which is a ancestor in the local version tree at distance at most $\dist$ from every version. 

\subsection{Browsing the Data Structure}
The key idea is that each node of \ds does not contain, by itself, enough information to determine what is the value of the fields of the node for a given version $v$. However, if we know that information for node $a$ at version $v$, and that we follow an edge to node $b$, we can deduce what is the value of the fields of $b$ in version $v$. Provided we know how to determine the value of the root of \ds in a given version $v$, we can then browse the data structure. 

\begin{figure}
\begin{center}
\includegraphics[scale=.8]{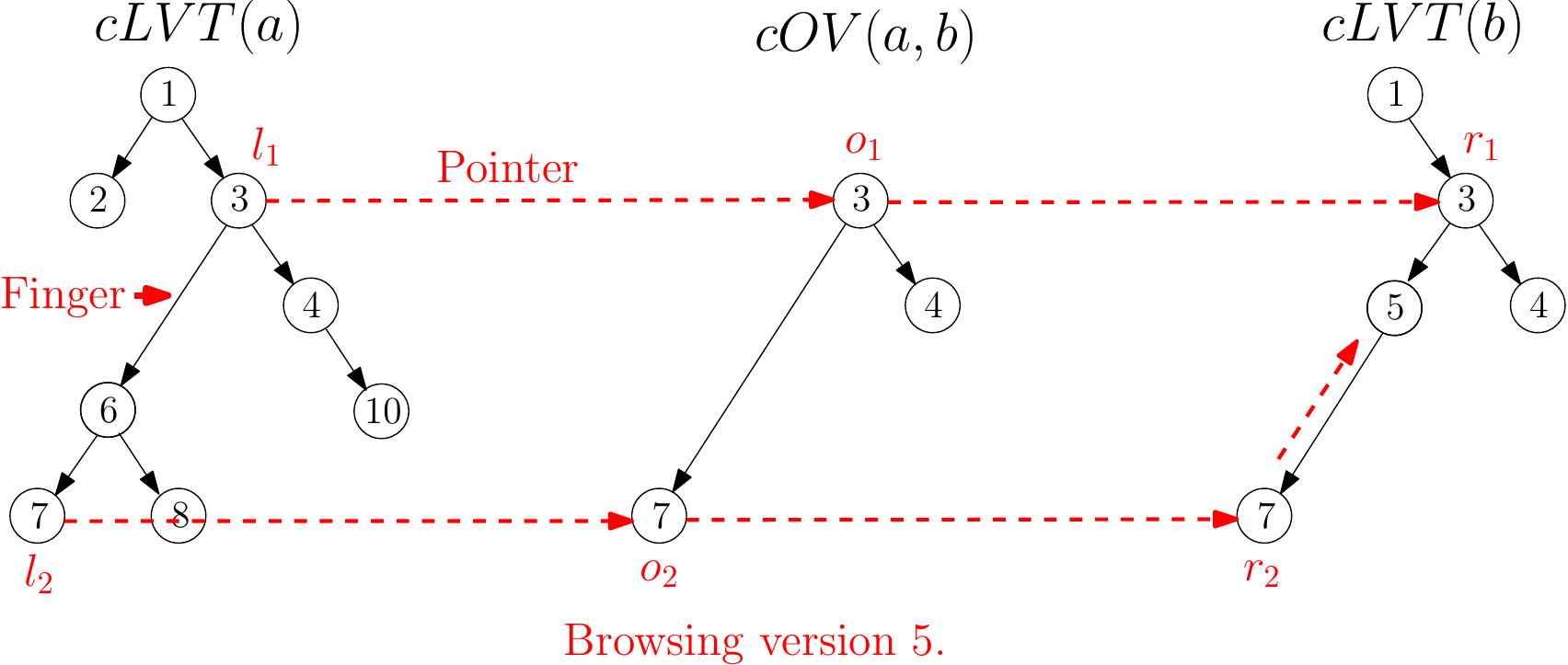}
\end{center}
\caption{\label{fig:browsing}To move the finger from a node $a$ to a node $b$, we use the edge overlay pointers. This illustrates the case with two disjoint pointers.}
\end{figure}

Suppose we are browsing version $v$; let $a$ be a node for which we know the representative $v'$ of $v$ in the compressed local version tree of $a$; thus all the fields of $a$ in version $v$ are immediately known. 
Suppose that one of these fields indicates there is an edge from $a$ to $b$ in version $v$, and that we wish to follow the edge so that we know the fields of $b$ in version $v$. This is the fundamental act of navigation. In \clvt{a} (see Figure~\ref{fig:browsing}), we find the closest vertex, on the path from the finger to the root of the local version tree, which explicitly exists in $\cov{a,b}$. This could be a senator, or any version in which the edge $(a,b)$ was modified. We denote this vertex by $l_{1}$. In the case where $v \not = l_1$, we also find the closest vertex which is a descendant of $v$  in the local version tree of $v$ that is present in \cov{a,b}, and denote it by $l_{2}$. If there is no such vertex below the finger, we take $l_{2}=l_{1}$. 	

Thus we now have two (possibly non-disjoint) vertices $l_{1}$ and $l_{2}$, which are in \clvt{a}, and which are also in \cov{a,b}. Following the version-edge pointers, we find vertices $o_{1}$ and $o_{2}$ in \cov{a,b}. Because of the invariant of edge overlays, we know that if $o_{1}$ and $o_{2}$ appear explicitly in \cov{a,b}, they also appear explicitly in \clvt{b}, and we find one or two corresponding vertices $r_{1}$ and $r_{2}$ in \clvt{b} using the edge-version pointers. 

It remains to place the finger on the vertex representing $v$ in $\clvt{b}$. Notice that, thanks to fractional cascading, the distance between $r_{1}$ and $r_{2}$ in $\clvt{b}$ is at most $\dist +1$. We first deal with the case where $r_{1}$ and $r_{2}$ are disjoint. Then $v$ is a vertex of the path between them in the compressed local version tree of $b$. We browse the path bottom-up, and using the version ID, the labels of the path are decreasing as we go up. We stop as soon as the ID of $v$ is larger or equal to that of the vertex we are visiting; this gives the representative of $v$ in the compressed local version tree of $b$, and thus all fields of $b$ in version $v$ are now accessible.

\begin{figure}
\begin{center}
\includegraphics[scale=.8]{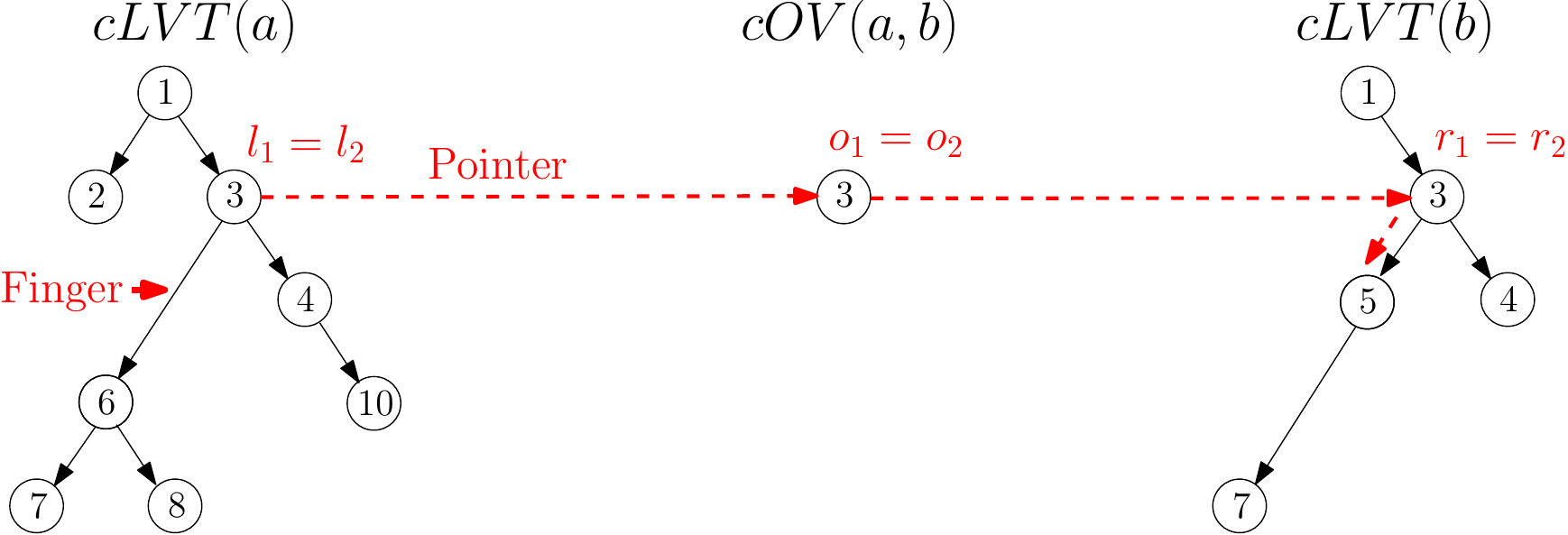}
\end{center}
\caption{\label{fig:browsing2}This illustrates the case where $r_{1}=r_{2}$. This would be the case if, for instance, the edge $(a,b)$ is created at version $3$, and is never deleted.}
\end{figure}

In the other case where the two vertices are identical (i.e., $r_{1}=r_{2}$, see Figure~\ref{fig:browsing2}), we know that the senator of $v$ is a vertex at distance less than $\dist$ of $r_{1}$ in \clvt{b}. If $r_{1}$ is a leaf of the contracted representation of $\clvt{b}$, we can place the finger on $r_{1}$. If not, we need to determine if $v$ belongs to the left or the right subtree of $r_{1}$. First notice that if $o_{1}$ was not a leaf of \cov{a,b}, then we know that $v$ should be placed on $r_{1}$, because otherwise we would have had two distinct pointers $r_{1}$ and~$r_{2}$. 

Thus, we can suppose that $o_{1}$ is a leaf of \cov{a,b}. This is where we use the forest $\for$ associated to $\gvd$: in $\for$, we cut (temporarily) the one or two edges connecting to the subtrees of version $r_{1}$. 

As $o_{1}$ is a leaf of $\cov{a,b}$, we know that, if a confluent merge happened in a version derived from $o_{1}$, the lowest common ancestor of the two versions merged cannot be in any subtree of $\cov{a,b}$ (for otherwise, we merged two versions in which both $a$ and $b$ existed, as the edge $a,b$ was valid). As no lowest common ancestor of any confluent merge can be part of a subtree of $r_{1}=o_{1}$, this means that all implicit and explicit versions in each subtree of $r_{1}$ in \lvt{b} belong to the same component of \for. 

Also, because of fractional cascading, we know that $r_{1}$ can only have a number of explicit descendants which is a function of $\dist$. For each edge of the subtree of $r_{1}$ in \clvt{b}, we can check if $v$ is an implicit node represented by that edge (by cutting temporarily in $\for$, in order to isolate the implicit nodes represented by that edge form the rest of the graph); we can also cut each of the explicit leaves and check if $v$ is represented by any of corresponding implicit subtrees. One of these cuts will isolate the edge or the vertex of $\clvt{b}$ on which the finger must be placed; we can check if an implicit version belongs to one of these cut parts in $O(log n)$ amortized time. Once we have found the appropriate location, we move the finger and all fields of $b$ in version $v$ are accessible.

In conclusion, we covered all cases, and starting from  a finger in \clvt{a}, we can place the finger accordingly in \clvt{b}, with a maximum of cuts in the link-cut tree proportional to the size of the $\dist$-neighborhood of the node $r_{1}$ in $\clvt{b}$. As $\dist$ is a constant (determined in the next section), the neighborhood has constant size.

We are still missing a final step: how to access to the root of version $v$? We can follow the pointer from the global version DAG, but we need to end up in a specific version. For each version created, we make sure that the pointer from the global version DAG points to that specific version in the node which is the root at version $v$, and we make that version explicit in that node. In other words, the node corresponding to the root of \ds for version $v$ will have an explicit vertex for $v$ in its local version tree.

\section{Analysis}
A method has been presented describing an algorithmic process to augment any data structure meeting our requirements to make it confluently persistent. In this section, we will analyze the time and space complexity of all these operations.
Recall that $n$ is the total number of atomic modifications performed so far.

\subsection{Creating a new version}

\begin{lemma}\label{lem:newversion}
Given a version $v$, deriving a new version takes $O(\log n)$ amortized time and space (either through a new branch or a confluent merge).
\end{lemma}

\begin{proof}
When we create a new version $v$, we need to perform the following operations:
\begin{itemize}
\item We add a version $v$ in \gvd and in $\for$, such that $v$ is the child of the version(s) that we are deriving from, in $O(\log n)$ time and $O(1)$ space.
\item We create an explicit version in \clvt{d}, where $d$ is the root of the data structure at version $v$. As we have a finger to the edge representing the previous version in $\clvt{d}$, the new vertex is inserted at a constant distance, and thus in $O(\log n)$ time. This new explicit version might be a senator, in which case we will have to propagate explicit versions in some neighbors of $d$. We will see in Lemma~\ref{lem:persistence} that this can be done in $O(\log n)$ amortized time and space. 
\item In the case of a confluent merge, we need to find a lowest common ancestor (in $O(log n)$ time), and cut two edges in $\for$, and create an edge between two nodes belonging to two different versions. Lemma~\ref{lem:newedge} proves that creating the edge can be done in $O(\log n)$ amortized time and space. 
\end{itemize}
\end{proof}

\begin{lemma}
Let $v \in \clvt{d}$ be a senatorial version. If one of its subtrees contains $i>0$ senators, it has size at least $3i/ 2$, provided $\dist \ge 2$. 
\end{lemma}

\begin{proof}
A node is a senator if and only if it has no senatorial ancestor at distance less than $\dist +1$. 
%Let us first consider the case where on every path from $v$ to the leaves of the subtree, there is at most one representative (out of $v$). 
The worst case, i.e., the smallest subtree achieving the largest number of senators versus non-senator ratio, is a complete balanced binary tree with $i$ senators in the leaves, connected to a path from $v$ of length $\dist+1 - \lceil\log_{2}{i} +1\rceil$. The complete balanced tree has size larger or equal to $2i-1$. 

So the total number of nodes in the subtree is at least $2i-1+\dist +1 - \lceil\log_{2}{i}+1\rceil $, which is always larger or equal to $3i/ 2$ provided $\dist \ge 2$.
\end{proof}

\begin{lemma}
Adding a new explicit version $v$ in the local version tree \clvt{d} of a node $d \in \ds$ takes $O(1)$ amortized time and space. \label{lem:persistence}
\end{lemma}

\begin{proof}
We use a potential function argument: every time a new explicit version is added to a node $d$, we increase the potential of \clvt{d} by a constant \const. Then, two cases occur. Either the new version will be senatorial or not. If it is not, we simply need to create a new vertex in \clvt{d}, and we use one unit of potential to pay for it, and overall the potential has increased by $\const - 1$.  If it is a senator, we need to visit all incoming edges $(d',d)$, add a new explicit version in each \cov{d',d}, and add a new explicit version in each \clvt{d'}. This means a total of $2\degbound$ extra insertions at most (as the in-degree is bounded by \degbound), that we need to be able to pay for using potential. 

At the time where the $i^{th}$ senator is inserted in any subtree of another senator, the subtree has size at least $3i / 2$. Thus, the potential is at least ${3i (\const-1) \over 2}- 2i \degbound$ after the insertion. As long as $ \const \ge {4 \degbound +1\over 3}$, we guarantee that there is always enough potential to pay for the insertions. 

The new explicit versions in \clvt{d'} might also be senators, in which case each of them will also require an extra $2\degbound$ insertions, also paid by potential, and so on. Thus, we conclude that an insertion takes $O(1)$ amortized time. 
\end{proof}

\subsection{Modifications to the structure}
\begin{lemma}
Creating or deleting an edge, or modifying any field of a node $d$ in the data structure at version $v$ takes $O(\log n)$ amortized time and space. \label{lem:newedge} 
\end{lemma}

\begin{proof}
Any modification will first imply to construct a new version, in $O(\log n)$ time (see Lemma~\ref{lem:newversion}). We are browsing version $v$, and we have a finger on an item that we want to change\footnote{Or rather, a pointer to the edge representing the previous version in \clvt{d}.}. To achieve this, we simply need to make the current version explicit in the local version tree of that node, and to propagate the information. If an edge is created or deleted, we need to change the info in both the predecessor and successor, and create an edge overlay containing a single version. All these operations are processed in $O(1)$ amortized time, as they consist in modifying a constant number of nodes, and adding a constant number of explicit versions (possibly needing amortized constant senatorial spreading). 
\end{proof}

\noindent From this set of lemmas, we deduce our main theorem:

\begin{theorem}
A data structure \ds with maximum in-degree $\maxdeg$ can be made confluently persistent using $O(U \log \maxdeg)$ space and time, where $U$ is the number of operations performed to create all versions of the structure. Each subsequent modification of the structure takes $O(\log n)$ amortized time and $O(1)$ space. Browsing the structure costs $O(\log \maxdeg \log n)$ amortized time per edge traversed. This holds provided no confluent merge operation uses twice the same node.
\end{theorem}

\begin{proof}
As pointed out in the model section, we can transform and maintain the structure so that the maximum in-degree is always bounded by $\degbound$, for any constant larger or equal to 2. 

By fixing $\dist =2$ and $\const = 3$, all lemmas hold. The fractional cascading and edge overlay invariants are maintained whenever explicit vertices are added.
Following an edge (d,d') only implies to look, in the $\clvt{d}$, for vertices in a $O(\dist)$-neighborhood, which is achieved in $O(\log n)$ time; and using a constant number of pointers, from \clvt{d} to \clvt{d'} using \cov{d,d'}. But our transformation for bounded in-degree implies a $O(\log maxdeg)$ factor on following a pointer. 
\end{proof}

In particular, notice that this means that in the model used in~\cite{FullPersistence} where the in-degree is bounded by a constant, confluent persistence can be achieved at a logarithmic cost (in the amortized asymptotic sense):

\begin{corollary}
A data structure \ds with constant in-degree can be made confluently persistent using $\Theta(U)$ space, where $U$ is the number of operations performed to create the structure. Each subsequent modification of the structure or new version created takes $O(\log n)$ amortized time and $O(1)$ space. Browsing the structure costs $O(\log n)$ amortized time per edge traversed. This holds provided no merge operation uses twice the same node.
\end{corollary}

{
\bibliographystyle{plain}
\bibliography{biblio}

\begin{thebibliography}{10}

\bibitem{persistenceCG}
A.~Bouroujerdi and B.~Moret.
\newblock {Persistence in Computational Geometry}.
\newblock In {\em Proc. 7th Canadian Conference in Computational Geometry},
  pages 241--246, 1995.

\bibitem{ConfluentDeques}
A.~Buchsbaum and R.~Tarjan.
\newblock Confluently persistent deques via data structuaral bootstrapping.
\newblock In {\em SODA '93: Proceedings of the fourth annual ACM-SIAM Symposium
  on Discrete algorithms}, pages 155--164, Philadelphia, PA, USA, 1993. Society
  for Industrial and Applied Mathematics.

\bibitem{FracCasc}
B.~Chazelle and L.~Guibas.
\newblock Fractional cascading: I. a data structuring technique.
\newblock {\em Algorithmica}, 1:133--162, 1986.

\bibitem{ConfluentTries}
E.~Demaine, S.~Langerman, and E.~Price.
\newblock Confluently persistent tries for efficient version control.
\newblock {\em Algorithmica}, 57(3):462--483, 2010.

\bibitem{ConfluentTries_Algorithmica}
E.D. Demaine, S.~Langerman, and E.~Price.
\newblock Confluently persistent tries for efficient version control.
\newblock {\em Algorithmica}, 57(3):462--483, July 2010.
\newblock Special issue of selected papers from 11th Scandinavian Workshop on
  Algorithm Theory, 2008.

\bibitem{PersistentArrays}
P.~Dietz.
\newblock Fully persistent arrays (extended array).
\newblock In {\em WADS '89: Proceedings of the Workshop on Algorithms and Data
  Structures}, pages 67--74, London, UK, 1989. Springer-Verlag.

\bibitem{FullPersistence}
J.~Driscoll, N.~Sarnak, D.~Sleator, and R.~Tarjan.
\newblock Making data structures persistent.
\newblock {\em Journal of Computer and System Sciences}, 38(1):86 -- 124, 1989.
\newblock Initially presented at STOC'86.

\bibitem{PersistentLists}
J.~Driscoll, D.~Sleator, and R.~Tarjan.
\newblock Fully persistent lists with catenation.
\newblock {\em J. ACM}, 41(5):943--959, 1994.

\bibitem{ConfluentPersistence}
A.~Fiat and H.~Kaplan.
\newblock Making data structures confluently persistent.
\newblock {\em J. Algorithms}, 48(1):16--58, 2003.

\bibitem{kaplan2009linear}
H.~Kaplan, N.~Rubin, and M.~Sharir.
\newblock {Linear data structures for fast ray-shooting amidst convex
  polyhedra}.
\newblock {\em Algorithmica}, 55(2):283--310, 2009.

\bibitem{PersistentLists2}
H.~Kaplan and R.~Tarjan.
\newblock Persistent lists with catenation via recursive slow-down.
\newblock In {\em STOC '95: Proceedings of the twenty-seventh annual ACM
  symposium on Theory of computing}, pages 93--102, New York, NY, USA, 1995.
  ACM.

\bibitem{DebugPersistence}
Z.~Liu.
\newblock A persistent runtime system using persistent data structures.
\newblock In {\em SAC '96: Proceedings of the 1996 ACM symposium on Applied
  Computing}, pages 429--436, New York, NY, USA, 1996. ACM.

\bibitem{PersistenceImplemented}
F.~Pluquet, S.~Langerman, A.~Marot, and R.~Wuyts.
\newblock Implementing partial persistence in object-oriented languages.
\newblock In {\em Proceedings of ALENEX'08}, 2008.

\bibitem{PointLocPersistence}
N.~Sarnak and R.~Tarjan.
\newblock Planar point location using persistent search trees.
\newblock {\em Commun. ACM}, 29(7):669--679, 1986.

\bibitem{LinkCutTree}
D.~Sleator and R.~Tarjan.
\newblock A data structure for dynamic trees.
\newblock In {\em Proceedings of the thirteenth annual ACM symposium on Theory
  of computing}, STOC '81, pages 114--122, 1981.

\end{thebibliography}
}
\end{document}